\documentclass[format=acmsmall, review=false]{acmart}
\usepackage{acm-ec-25}
\usepackage{booktabs} 
\usepackage[ruled,vlined,linesnumbered]{algorithm2e}

\SetAlFnt{\small}
\SetAlCapFnt{\small}
\SetAlCapNameFnt{\small}
\SetAlCapHSkip{0pt}
\IncMargin{-\parindent}

\setcitestyle{authoryear}

\usepackage{amsmath,amsthm,mathtools}

\usepackage{natbib}

\usepackage{multirow}

\usepackage{appendix}
\usepackage{nicefrac}

\usepackage[utf8]{inputenc} 
\usepackage[T1]{fontenc}    
\usepackage{hyperref}       
\usepackage{url}            
\usepackage{booktabs}       
\usepackage{amsfonts}       
\usepackage{nicefrac}       
\usepackage{microtype}      
\usepackage{array}

\usepackage{setspace}
\SetKwComment{Comment}{//}{}

\usepackage[capitalise,nameinlink]{cleveref}

\definecolor{links}{RGB}{11, 85, 255}
\definecolor{cites}{RGB}{0, 200, 0}
\definecolor{urls}{RGB}{255, 116, 0}
\hypersetup{colorlinks={true},linkcolor={links},citecolor=[named]{cites},urlcolor={urls}}

\newcommand{\calD}{\mathcal{D}}

\newcommand{\A}{\mathcal{A}}
\newcommand{\EE}{\mathcal{E}}
\newcommand{\bin}{\mathcal{B}}

\newcommand{\disc}{\text{DISC}}
\newcommand{\cd}{\text{CD}}
\newcommand{\ef}{\text{EF}}
\newcommand{\prop}{\text{PROP}}
\newcommand{\collection}{\mathcal{S}}
    
\newtheorem{theorem}{Theorem}
\newtheorem{lemma}{Lemma}
\newtheorem{claim}[lemma]{Claim}

\newtheorem{definition}{Definition}

\let\originalleft\left
\let\originalright\right
\renewcommand{\left}{\mathopen{}\mathclose\bgroup\originalleft}
\renewcommand{\right}{\aftergroup\egroup\originalright}
\renewcommand{\epsilon}{\varepsilon}
\newcommand{\eps}{\varepsilon}

\allowdisplaybreaks

\title{A new lower bound for multi-color discrepancy with applications to fair division}

\author{Ioannis Caragiannis}
\affiliation{%
  \institution{Department of Computer Science, Aarhus University}
  \streetaddress{{\AA}bogade 34}
  \city{Aarhus N}
  \postcode{8200}
  \country{Denmark}
}
\author{Kasper Green Larsen}
\affiliation{%
  \institution{Department of Computer Science, Aarhus University}
  \streetaddress{{\AA}bogade 34}
  \city{Aarhus N}
  \postcode{8200}
  \country{Denmark}
}
\author{Sudarshan Shyam}
\affiliation{%
  \institution{Department of Computer Science, Aarhus University}
  \streetaddress{{\AA}bogade 34}
  \city{Aarhus N}
  \postcode{8200}
  \country{Denmark}
}

\begin{document}

\begin{abstract}
A classical problem in combinatorics seeks colorings of low discrepancy. More concretely, the goal is to color the elements of a set system so that the number of appearances of any color among the elements in each set is as balanced as possible. We present a new lower bound for multi-color discrepancy, showing that there is a set system with $n$ subsets over a set of elements in which any $k$-coloring of the elements has discrepancy at least $\Omega\left(\sqrt{\frac{n}{\ln{k}}}\right)$. This result improves the previously best-known lower bound of $\Omega\left(\sqrt{\frac{n}{k}}\right)$ of~\citet{DS03} and may have several applications. Here, we explore its implications on the feasibility of fair division concepts for instances with $n$ agents having valuations for a set of indivisible items. The first such concept is known as consensus $1/k$-division up to $d$ items (\cd$d$) and aims to allocate the items into $k$ bundles so that no matter which bundle each agent is assigned to, the allocation is envy-free up to $d$ items. The above lower bound implies that \cd$d$ can be infeasible for $d\in \Omega\left(\sqrt{\frac{n}{\ln{k}}}\right)$. We furthermore extend our proof technique to show that there exist instances of the problem of allocating indivisible items to $k$ groups of $n$ agents in total so that envy-freeness and proportionality up to $d$ items are infeasible for $d\in \Omega\left(\sqrt{\frac{n}{k\ln{k}}}\right)$ and $d\in \Omega\left(\sqrt{\frac{n}{k^3\ln{k}}}\right)$, respectively. The  lower bounds for fair division improve the currently best-known ones by  \citet{MS22}.
\end{abstract}

\maketitle
\setcounter{page}{1}

\section{Introduction}
Allocating indivisible items to agents with valuations for them has been a key problem in {\em fair division}. The notion of envy-freeness up to one item (EF1), introduced by~\citet{B11} is a well-established fairness notion today. An allocation of items to agents is EF1 if every agent (weakly) prefers the bundle of items allocated to her to the bundle of items allocated to any other agent after removing one item from the latter. In contrast to the notion of envy-freeness, which is very demanding for indivisible items, EF1 can always be achieved in a number of different ways: via the envy-cycle elimination algorithm of~\citet{LMMS04}, the folklore round-robin algorithm, while it is compatible with Pareto-optimality~\citep{CKMPSW19}.

A natural generalization of the standard fair division setting with indivisible items assumes that agents with different valuations for a set of items are partitioned into groups. In this setting, an allocation has one bundle per group, and the agent gets value for the items in the bundle allocated to her group. Unfortunately, EF1 allocations may not exist in this setting~\citep{KSV20}. Actually, even the further relaxed notion of envy-freeness up to $d$ items (EF$d$) may be infeasible even when $d$ is a function of the number of agents and the number of groups. \citet{MS22} show that even envy-freeness up to $\Omega\left(\sqrt{n}/{k^2}\right)$ items can be infeasible for instances with $n$ agents partitioned into $k$ groups. Non-trivial positive results are also known; EF$d$ allocations do exist for $d\in O\left(\sqrt{n}\right)$ in all instances with $n$ agents. Notice that this bound does not depend on the number of groups.

Other fairness properties that have been considered for groups of agents include relaxations of proportionality. An allocation of items to agents partitioned into $k$ groups is proportional up to $d$ items (PROP$d$) if the value that each agent has for the bundle of items allocated to her group together with the $d$ most valuable items not allocated to her group is at least $1/k$ times her total value for all items. The simplest version of PROP$1$ was introduced by~\citet{CFS17} and is considered for groups of agents by~\citet{MS22}, who prove similar bounds on $d$ for the feasibility of PROP$d$ allocations with those for EF$d$ mentioned above. Another relevant fairness notion is known as consensus $1/k$-division up to $d$ items (CD$d$). Here, a partition of the items to $k$ bundles is CD$d$ for a set of $n$ agents with valuations for the items if the values an agent has for any pair of items differ by at most $d$. \citet{MS22} prove that CD$d$ partitions always exist for $d\in O\left(\sqrt{n}\right)$ and may not exist for $d\in \Omega\left(\sqrt{n/k}\right)$.  

In their proofs, \citet{MS22} exploit several problems and statements from {\em discrepancy theory}. For example, their lower bound for \cd$d$ follows after establishing a connection between $k$-allocations and $k$-colorings of set systems. Given a set system consisting of a universe of elements and a collection of $n$ element subsets, a $k$-coloring of the elements has a discrepancy $d$ if the number of elements colored with any color in each subset $S$ is between $|S|/k-d$ and $|S|/k+d$. Intuitively, we may think of the elements of a set system as the items of a corresponding fair division instance. Each set corresponds to a distinct agent with valuation $1$ for each item corresponding to an element in her set and valuation $0$ for any other item. A $k$-coloring of the elements directly defines an allocation of the items into $k$ bundles. Now, the relation between the minimum discrepancy among all $k$-colorings of the set system and the minimum value of $d$ for which a \cd$d$ allocation exists in the corresponding fair division instance should be clear. The lower bound of~\citet{MS22} on \cd$d$ exploits such a relation and follows directly by a lower bound on multi-color discrepancy due to~\citet{DS03}. Other notions, such as the weighted discrepancy of $2$-colorings, are used by~\citet{MS22} to get upper and lower bounds on \ef$d$ and \prop$d$.

\subsection{Our contribution}
In this paper, we improve the bounds on $d$ for which CD$d$ partitions and EF$d$ and PROP$d$ allocations may not exist. Our new bounds are $\Omega\left(\sqrt{\frac{n}{\ln{k}}}\right)$, $\Omega\left(\sqrt{\frac{n}{k\ln{k}}}\right)$, and $\Omega\left(\sqrt{\frac{n}{k^3\ln{k}}}\right)$ and apply even to instances with binary agent valuations. For instances with $n_h$ agents in group $h \in [k]$, we also have another lower bound of $ d \in \Omega \left( \sqrt{\frac{\min\{n_1,n_2,...,n_k\}}{\ln k}} \right)$ for \prop$d$. For CD$d$, we exploit the relation of CD$d$ partitions on instances with binary valuations with $k$-colorings in set systems that have discrepancy $d$. So, our new lower bound for CD$d$ follows directly by a new lower bound of $\Omega\left(\sqrt{\frac{n}{\ln{k}}}\right)$ on multi-color discrepancy. This improves the twenty-year-old bound of $\Omega\left(\sqrt{n/k}\right)$ by~\citet{DS03} and may have applications to other areas as well. In our proof, we define a probability distribution over set systems with $n$ subsets over an appropriately defined collection of elements and show that the probability that no $k$-coloring has discrepancy at most $d$ in all sets of the set system returned by the distribution is less than $1$. This probabilistic argument implies that there exists a set system (one of those in the support of the distribution) for which any $k$-coloring has discrepancy more than $d$.

In contrast to the approach of~\citet{MS22}, our lower bounds for EF$d$ and PROP$d$ do not follow by applying bounds from discrepancy theory as black boxes. Instead, adapting our construction and proof for our multi-color discrepancy lower bound, we prove that there exist instances with $n$ agents partitioned into $k$ groups and having binary valuations for a set of items, so that any allocation of the items to the $k$ groups is not EF$d$ or PROP$d$ for the values of $d$ claimed above.

\subsection{Further related work}
Existing work on fair division concepts for groups of agents has focused on relaxations of envy-freeness. The results of~\citet{MS22} on EF$d$ improve considerably previous ones by~\citet{KSV20} and~\citet{SS19}. \citet{MS17} study EF$d$ for groups of agents, assuming that the values are drawn independently from a common probability distribution. Consensus $1/k$-division up to $d$ items generalizes and relaxes the consensus halving problem (e.g., see~\citet{A87,SS03}) from two groups and divisible items to multiple groups and indivisible items. Different fairness notions that involve groups of agents in their definitions are studied by~\citet{CFSV19} and~\citet{AR20}.

The rich literature on discrepancy theory is the topic of several books, e.g., see~\cite{C00,M99,CST14}. Classical bounds on the discrepancy of $2$-colorings of set systems follow by the seminal work of~\citet{AS00} and~\citet{S85}. Extensions to multi-colorings are considered by~\citet{DS03}. The problem is equivalent to coloring the nodes of a hypergraph so that all colors are used approximately the same number of times in each hyperedge. In the discrepancy theory literature, bounds on $2$- or multi-color discrepancy are expressed in terms of the number of nodes, the number of hyperedges, and the number of colors. Here, we aim to bound the minimum multi-color discrepancy using only the number of sets (i.e., hyperedges) and colors as parameters. \citet{MS22} explain how existing discrepancy theory results can be expressed under this terminology, yielding the currently best upper and lower bounds of $O(\sqrt{n})$ and $\Omega\left(\sqrt{n/k}\right)$, respectively.

\subsection{Roadmap}
The rest of the paper is structured as follows. We present formal definitions of the notions we study in Section~\ref{sec:prelim}. The lower bound for multi-color discrepancy and its implication to consensus $1/k$-division up to $d$ items is presented in Section~\ref{sec:disc}. The lower bounds for EF$d$ and PROP$d$ are presented in Section~\ref{sec:ef-prop}. We conclude in Section~\ref{sec:open}. 

\section{Preliminaries}\label{sec:prelim}
We consider fair division settings with a set of $n$ agents with (additive) valuations for $m$ indivisible items. Using the notation $[t]=\{1, 2, ..., t\}$ for integer $t\geq 1$, we identify both agents and items as positive integers in $[n]$ and $[m]$, respectively. Each agent $i\in [n]$ has a non-negative valuation $v_i(j)$ for item $j\in [m]$. The valuation of an agent for a set of items $S\subseteq [m]$ is then simply $v_i(S)=\sum_{j\in S}{v_i(j)}$.

Given an integer $k\geq 2$, a $k$-allocation $A=(A_1, A_2, ..., A_k)$, is simply an ordered partition of the items into $k$ bundles. We consider three fairness notions. The first one, called consensus $1/k$-division up to $d$ items, is defined as follows.

\begin{definition}[consensus $1/k$-division up to $d$ items]
    Given a set of $n$ agents with valuations for a set of items and an integer $k\geq 2$, a $k$-allocation $A=(A_1, ..., A_k)$ of the items to $k$ bundles is a consensus $1/k$-division up to $d$ items (or $\emph{CD}d$, for short) if for every agent $i\in [n]$ and every pair of integers $h$ and $\ell$ from $[k]$, there is a set $B$ of at most $d$ items from bundle $A_\ell$ so that $v_i(A_h)\geq v_i(A_\ell\setminus B)$.
\end{definition}

We are interested in the minimum value of $d$ so that CD$d$ $k$-allocations exist for all instances with $n$ agents.

\begin{definition}
Given parameters $n$ and $k$, we denote by $\emph{\cd}(n,k)$ the minimum value of $d$ so that for any instances with $n$ agents having valuations for a set of items, there is a \emph{\cd}$d$ $k$-allocation.
\end{definition}

Our next two fairness notions extend well-known relaxations of envy-freeness and proportionality to groups of agents. We consider settings in which $n$ agents are partitioned into $k$ groups. Again, we identify the groups by positive integers in $[k]$. We will denote by $n_h$ the number of agents in group $h\in [k]$ and by $g(i)$ the group to which agent $i\in [n]$ belongs.

\begin{definition}[envy-freeness up to $d$ items]
    Given a set of $n$ agents partitioned into $k$ groups, a $k$-allocation is {\em envy-free up to $d$ items} (or \emph{EF}$d$ for short) if for every agent $i$ and group $h\in [k]$, there is a set $B$ of at most $d$ items so that $v_i(A_{g(i)})\geq v_i(A_h\setminus B)$.
\end{definition}

\begin{definition}[proportionality up to $d$ items]
    Given a set of $n$ agents partitioned into $k$ groups, a $k$-allocation is {\em proportional up to $d$ items} (or \emph{PROP}$d$, for short) if for every agent $i$, there is a set $B$ of at most $d$ items not allocated to group $g(i)$ (i.e., $B\subseteq [m]\setminus A_{g(i)}$) so that $v_i(A_{g(i)})\geq \frac{1}{k}\cdot \sum_{g\in [m]}{v_i(g)}-v_i(B)$.
\end{definition}

Again, we are interested in the minimum value of $d$ so that EF$d$ and PROP$d$ $k$-allocations exist for all instances with $n$ agents partitioned into $k$ groups. 

\begin{definition}
Given $k\geq 2$ positive integer parameters $n_1, n_2, ..., n_k$, we denote by \emph{\ef}$(n_1, ..., n_k)$ (respectively, \emph{\prop}$(n_1, ..., n_k)$) the minimum value of $d$ so that, for any instances with $n_h$ agents in group $h\in [k]$ with valuations for a set of items, there exist an \emph{EF}$d$ (respectively, \emph{\prop}$d$) $k$-allocation.
\end{definition}

The notion of \cd$d$ is strongly related to a notion from discrepancy theory. In the multi-color discrepancy problem, we are given a set system $(U,\collection)$ consisting of a universe of elements $U$ and a collection $\collection$  with $n$ subsets $S_1$, $S_2$, ..., $S_n$ of $U$. A $k$-coloring $\chi$ of the set system $(U,\collection)$ is simply an assignment of colors from $[k]$ to the elements of $U$, with $\chi(s)$ denoting the color given to the element $s\in U$ by $\chi$. Given a $k$-coloring $\chi$ for the elements of a universe $U$ and a color $h\in [k]$, we denote by $\chi^{-1}(h)$ the set of elements that are colored with $h$ under $\chi$. We are interested in studying the discrepancy of $k$-colorings defined as follows.

\begin{definition}[multi-color discrepancy]
    A $k$-coloring of a set system $(U,\collection)$ with a universe of elements $U$ and a collection $\collection$ of $n$ sets has discrepancy $d$ if 
    \[\max_{h\in [k]}\max_{i\in [n]}{\left||\chi^{-1}(h)|-\frac{|S_i|}{k}\right|} \leq d.\]
\end{definition}
Similarly to the fairness notions above, we are interested in the minimum value of $d$ for which $k$-colorings of discrepancy $d$ always exist.

\begin{definition}
    Given integers $k\geq 2$ and $n$, we denote by \emph{\disc}$(n,k)$ the minimum value of $d$ so that all set systems with a collection consisting of $n$ element subsets have a $k$-coloring of discrepancy $d$.
\end{definition}

In the next two sections, we present new lower bounds for the quantities \disc$(n,k)$, \cd$(n,k)$, \ef$(n_1, ..., n_k)$, and \prop$(n_1, ..., n_k)$ defined above. Our results are summarized in Table~\ref{tab:results}, together with a comparison with the previous best-known lower bounds from the literature.

\setlength{\extrarowheight}{10pt}
\begin{table}[h]
    \centering\small
    \begin{tabular}{|c|c|c|l|}\hline
    quantity & our bound & previous bound & reference\\\hline
         $\disc(n,k)$ & $\Omega\left(\sqrt{\frac{n}{\ln{k}}}\right)$ & $\Omega\left(\sqrt{\frac{n}{k}}\right)$ & \cite{DS03}\\
         $\cd(n,k)$ & $\Omega\left(\sqrt{\frac{n}{\ln{k}}}\right)$ & $\Omega\left(\sqrt{\frac{n}{k}}\right)$ & \cite{MS22}\\
         $\ef(n_1, ..., n_k)$ & $\Omega\left(\sqrt{\frac{n}{k\ln{k}}}\right)$ & $\Omega\left(\sqrt{\frac{\max\{n_1, ..., n_k\}}{k^3}}\right)$ & \cite{MS22}\\
         $\prop(n_1, ..., n_k)$ & \begin{tabular}{@{}c@{}}$\Omega\left(\sqrt{\frac{n}{k^3\ln{k}}}\right)$ \\ $\Omega \left( \sqrt{\frac{\min\{n_1,n_2,...,n_k\}}{\ln k}} \right)$\end{tabular} & $\Omega\left(\sqrt{\frac{\max\{n_1, ..., n_k\}}{k^3}}\right)$ & \cite{MS22}
         \\\hline
    \end{tabular}
    \caption{Our lower bounds compared to previous work. For \disc, $n$ and $k$ denote the number of sets and colors, respectively. For the remaining quantities, $n$ and $k$ denote the number of agents and groups/bundles, respectively. For \ef\ and \prop, $n_h$ denotes the number of agents in the $h$-th group; it is $n=n_1+...+n_k$. Notice that $\max\{n_1,..., n_k\}$ can be as low as $n/k$. So, the previous lower bounds for \ef\ and \prop\ can be as low as $\Omega\left(\sqrt{n}/k^2\right)$. Also, $\min\{n_1,n_2, ..., n_k\}$ can be as high as $n/k$, so our second lower bound for \prop\ can be as high as $\Omega\left(\sqrt{\frac{n}{k\ln{k}}}\right)$.}
    \label{tab:results}
\end{table}

In our proofs, we use extensively an anti-concentration bound for the binomial probability distribution $\bin(t,1/2)$ with $t$ trials with a success probability of $1/2$ per trial. In other words, we focus on random variables defined as the sum $\sum_{j=1}^{t}{X_i}$ of $t$ independent and identically distributed Bernoulli random variables $X_1$, $X_2$, ..., $X_t$ with $\Pr[X_i=1]=1/2$ for $i\in [t]$. The following lemma has been proved by~\citet{KY15} and bounds from below the probability that the random variable is considerably smaller or considerably larger than its expectation $t/2$.

\begin{lemma}[Reverse Chernoff bound, e.g., see~\cite{KY15}]\label{lem:reverse-Chernoff}
Let $X\sim \bin(t,1/2)$. Then, for every $\eps\in (0,1/2]$ so that $\eps^2t\geq 6$, it holds
\begin{align*}
        \Pr\left[X\leq \frac{t}{2}\cdot (1-\eps)\right] &\geq \exp\left(-\frac{9\eps^2t}{2}\right)
\end{align*}
and 
\begin{align*}
        \Pr\left[X\geq \frac{t}{2}\cdot (1+\eps)\right] &\geq \exp\left(-\frac{9\eps^2t}{2}\right).
\end{align*}
\end{lemma}

\section{A lower bound for multi-color discrepancy}\label{sec:disc}
We devote this section to proving our lower bound for $\disc(n,k)$.
\begin{theorem}\label{thm:discr}
Let $k\geq 3  + 6e^{48}$ and $n\geq 1+147\cdot e^{48} \ln{k}$ be integers.
    There exists a set system with $n$ sets over a set of elements, so that any $k$-coloring of its elements has discrepancy at least $\Omega\left(\sqrt{\frac{n}{\ln{k}}}\right)$.
\end{theorem}

We will prove Theorem~\ref{thm:discr} using a probabilistic argument.  For appropriate parameters $m$ and $d$, we will define a random set system consisting of $n$ subsets of $m$ elements,\footnote{We remark that, in our proof, we select the number of elements so that our lower bound on $d$ as a function of $n$ and $k$ is as high as possible. If we are further constrained by the number of elements $m$, we can slightly modify our proof to get a discrepancy lower bound of $\Omega\left(\sqrt{\frac{m}{k}\ln{\frac{nk}{m\ln{k}}}}\right)$.} such that the probability that there is some $k$-coloring with discrepancy at most $d$ is strictly smaller than $1$. This will imply that there exists a set system for which any $k$-coloring of the elements has discrepancy higher than $d$, proving the desired lower bound.

\begin{proof}
We will use the construction below with $m=\left\lfloor \frac{(n-1)k}{3 e^{48}\ln{k}}\right\rfloor$ and $d=\sqrt{\frac{m}{k}}$.  Together with the restriction $n\geq 1+147\cdot e^{48}\ln{k}$, these definitions imply that $d\in \Omega\left(\sqrt{\frac{n}{\ln{k}}}\right)$.  Furthermore,  notice that $\sqrt{m/k} = \sqrt{\frac{1}{k} \left\lfloor \frac{(n-1)k}{3 e^{48}\ln{k}}\right\rfloor} \geq \sqrt{\frac{1}{k} \left\lfloor 49k\right\rfloor} = 7$ and, thus, $d\leq \frac{m}{7k}$; this property will be useful later in the proof.

The construction is as follows.  There is a set $S_0$ that includes all elements and $n-1$ sets $S_1$, $S_2$, ..., $S_{n-1}$ defined in the following way: for $i\in [n-1]$ and $j\in[m]$,  element $j$ is included in set $S_i$ with probability $1/2$.  All random events are independent.

We denote by $X$ the set of $k$-colorings $\chi$ satisfying
\begin{align}\label{eq:disc-bounds}    
    \frac{m}{k}-d &\leq \left|\chi^{-1}(h)\right|\leq \frac{m}{k}+d
\end{align}
for every $h\in [k]$.  Clearly, a $k$-coloring that does not belong to $X$ has discrepancy higher than $d$ for some color and set $S_0$, since it must be 
\begin{align*}
    \left| \left|\chi^{-1}(h) \cap S_0 \right| - \frac{|S_0|}{k}\right| = \left| \left|\chi^{-1}(h) \right| - \frac{m}{k}\right| > d
\end{align*}
 for some color $h \in [k]$.

Consider a $k$-coloring $\chi\in X$. We will show that the probability that $\chi$ has discrepancy at most $d$ for the random set system above is less than $k^{-m}$. As there are at most $k^m$ different colorings, this implies that the probability that none of them has discrepancy $d$ is positive. This is sufficient to prove that a set system with no $k$-coloring of discrepancy at most $d$ exists, completing the proof of the theorem.

For $i \in [n-1]$ and $h \in \left\{1,2,...,\lfloor k/2 \rfloor\right\}$ denote by $\EE_{i,h}$ the event defined as
\begin{align*}
    \left| \chi^{-1}(h) \cap S_i\right| < \frac{m}{2k} - d.
\end{align*}
Similarly, for $i \in [n-1]$ and $h \in \left\{\lfloor k/2 \rfloor+1,...,k\right\}$, denote by $\EE_{i,h}$ the event defined as 
\begin{align*}
    \left| \chi^{-1}(h) \cap S_i\right| > \frac{m}{2k}+ d.  
\end{align*}

The next lemma provides a sufficient condition so that combinations of events $\EE_{i,h}$ yield high discrepancy for coloring $\chi$.
\begin{lemma} \label{lem:badcoloring}
    If there exists $i \in [n-1]$, $h_1 \in \{1,2,...,\lfloor k/2 \rfloor\}$, and $h_2 \in \{\lfloor k/2 \rfloor+1,...,k\}$ such that both events $\EE_{i,h_1}$ and $\EE_{i,h_2}$ are true, the discrepancy of coloring $\chi$ is greater than $d$.
\end{lemma}

\begin{proof}
    Notice that if the events $\EE_{i,h_1}$ and $\EE_{i,h_2}$ are true, we have $\left| \chi^{-1}(h_1) \cap S_i\right| < \frac{m}{2k}- d$ and $\left| \chi^{-1}(h_2) \cap S_i\right| > \frac{m}{2k}+ d$, respectively. 
Now, if $\frac{|S_i|}{k} > \frac{m}{2k}$, we have
    \begin{align*}
      \left|  \left| \chi^{-1}(h_1) \right| - \frac{|S_i|}{k} \right| = \frac{|S_i|}{k} - \left| \chi^{-1}(h_1) \right| > d.
    \end{align*}
Otherwise, if $\frac{|S_i|}{k} \leq \frac{m}{2k}$, we have
    \begin{align*}
      \left|  \left| \chi^{-1}(h_2) \right| - \frac{|S_i|}{k} \right| = \left| \chi^{-1}(h_2) \right| - \frac{|S_i|}{k} > d.
    \end{align*}
    Both cases imply a discrepancy higher than $d$ for coloring $\chi$.
\end{proof}

We now bound the probabilities of the events defined above. Specifically, we prove the following lemma.

\begin{lemma} \label{lemma:actualbound}
    For every $i \in [n-1]$, $h \in [k]$ it holds 
    \begin{align*}
        \Pr\left[\overline{\EE_{i,h_1}}\right] <\exp(-\exp(-48)).
    \end{align*}
\end{lemma}

\begin{proof}
    Notice that the random variable $\left| \chi^{-1}(h) \cup S_i \right|$ follows the binomial probability distribution $\bin(|\chi^{-1}(h)|, 1/2)$ with $|\chi^{-1}(h)|$ trials and success probability $1/2$ per trial.
    For $i \in [n-1]$ and $h \in \{1,2,...,\lfloor k/2 \rfloor\}$, we have 
    \begin{align}
        \nonumber
        \Pr[\EE_{i,h}]  &= \Pr\left[\left|\chi^{-1}(h) \cup S_i\right| < \frac{m}{2k}-d\right] \\\nonumber
        &\geq \Pr\left[\left|\chi^{-1}(h) \cup S_i\right| < \frac{\left|\chi^{-1}(h)\right|}{2} - \frac{3d}{2}\right] \\
        \label{equation:event1}
        &\geq \exp \left(-\frac{81}{2}\cdot \frac{d^2}{\left|\chi^{-1}(h)\right|} \right).
    \end{align}

    For $i \in [n-1]$ and $h \in \{\lfloor k/2 \rfloor+1,...,k\}$, we have 
    \begin{align}
        \nonumber
        \Pr[\EE_{i,h}] &= \Pr\left[\left|\chi^{-1}(h) \cup S_i\right| > \frac{m}{2k}+d\right] \\\nonumber
        &\geq \Pr\left[\left|\chi^{-1}(h) \cup S_i\right| > \frac{\left|\chi^{-1}(h)\right|}{2} - \frac{3d}{2}\right] \\
        \label{equation:event2}
        &\geq \exp \left(-\frac{81}{2} \cdot \frac{d^2}{\left|\chi^{-1}(h)\right|} \right).
    \end{align}
The first inequality in the derivations of (\ref{equation:event1}) and (\ref{equation:event2}) follow by the right and left inequality in (\ref{eq:disc-bounds}) respectively. The second inequality follows by applying the reverse Chernoff bound (Lemma \ref{lem:reverse-Chernoff}) with $t = |\chi^{-1}(h)|$ and $\eps = \frac{3d}{|\chi^{-1}(h)|}$. The next claim justifies that we can indeed do so.

\begin{claim}
For $\eps=\frac{3d}{|\chi^{-1}(h)|}$ and $t=|\chi^{-1}(h)|$, it holds that $\eps \leq 1/2$ and $\eps^2\cdot |\chi^{-1}(h)| \geq 6$.
\end{claim}

\begin{proof}
By the fact that $d\leq \frac{m}{7k}$ and the left inequality in (\ref{eq:disc-bounds}), we have $6d\leq \frac{m}{k}-d\leq |\chi^{-1}(h)|$, which implies that $\eps =\frac{3d}{|\chi^{-1}(h)|}\leq \frac{1}{2}$. Now, notice that the facts $d=\sqrt{\frac{m}{k}}$ and $d\leq \frac{m}{7k}$ imply that $d\leq \frac{d^2}{7}\leq \frac{d^2}{2}$. Using this observation and the right inequality in (\ref{eq:disc-bounds}), we get $|\chi^{-1}(h)| \leq \frac{m}{k}-d=d^2+d\leq \frac{3d^2}{2}$. Hence, $\eps^2 t=\frac{9d^2}{|\chi^{-1}(h)|}\geq 6$. The claim follows.
\end{proof}

Now, using the left inequality in (\ref{eq:disc-bounds}), equations (\ref{equation:event1}) and (\ref{equation:event2}) yield
\begin{align*}
    \Pr[\EE_{i,h}] &\geq \exp \left(- \frac{81}{2} \cdot \frac{d^2}{\frac{m}{k}-d}\right) \geq \exp \left(- \frac{81}{2} \cdot \frac{7d^2k}{6m} \right) 
    > \exp \left(-48\right).
\end{align*}
Finally, by the inequality $1-z \leq e^{-z}$ for every real $z$, we get
\begin{align*}
    \Pr \left[ \overline{\EE_{i,h}} \right] &= 1 - \Pr \left[ \EE_{i,h} \right]< 1 - \exp(-48) \leq \exp(-\exp(-48)),
\end{align*}
as desired.
\end{proof}

We are now ready to complete the proof of Theorem~\ref{thm:discr}. Denote by $\calD$ the event that coloring $\chi \in X$ has discrepancy at most $d$ for the random set system. By Lemma \ref{lem:badcoloring}, we have
\begin{align}
\nonumber
    \calD &\subseteq \overline{\bigvee_{i \in [n-1]} \left(  \left(\bigvee_{h_1 \in \{1,...,\lfloor k/2 \rfloor\}}  \EE_{i,h_1} \right) \wedge \left(\bigvee_{h_2 \in \{\lfloor k/2 \rfloor+1,...,k\}} \EE_{i,h_2} \right) \right)} \\\nonumber
    &= \bigwedge_{i \in [n-1]}  \left( \overline{  \left(\bigvee_{h_1 \in \{1,...,\lfloor k/2 \rfloor\}}  \EE_{i,h_1} \right) \wedge \left(\bigvee_{h_2 \in \{\lfloor k/2 \rfloor+1,...,k\}} \EE_{i,h_2}  \right)} \right) \\\nonumber
    &= \bigwedge_{i \in [n-1]} \left( \left(  \overline{\bigvee_{h_1 \in \{1,...,\lfloor k/2 \rfloor\}} \EE_{i,h_1}} \right) \vee \left( \overline{\bigvee_{h_2 \in \{\lfloor k/2 \rfloor+1,...,k\}} \EE_{i,h_2}} \right) \right) \\\label{equation:boundbadevent}
    &= \bigwedge_{i \in [n-1]} \left( \left(  \bigwedge_{h_1 \in \{1,...,\lfloor k/2 \rfloor\}} \overline{ \EE_{i,h_1}} \right) \vee \left( \bigwedge_{h_2 \in \{\lfloor k/2 \rfloor+1,...,k\}} \overline{ \EE_{i,h_2}} \right) \right)
\end{align}
Using equation (\ref{equation:boundbadevent}) and Lemma \ref{lem:badcoloring}, we obtain
\begin{align}\nonumber
    \Pr[\calD] &\leq \Pr \left[\bigwedge_{i \in [n-1]} \left( \left(  \bigwedge_{h_1 \in \{1,...,\lfloor k/2 \rfloor\}} \overline{ \EE_{i,h_1}} \right) \vee \left( \bigwedge_{h_2 \in \{\lfloor k/2 \rfloor+1,...,k\}} \overline{ \EE_{i,h_2}} \right) \right) \right] \\\nonumber
    &= \prod_{i \in [n-1]}  \Pr \left [ \left( \left(  \bigwedge_{h_1 \in \{1,...,\lfloor k/2 \rfloor\}} \overline{ \EE_{i,h_1}} \right) \vee \left( \bigwedge_{h_2 \in \{\lfloor k/2 \rfloor+1,...,k\}} \overline{ \EE_{i,h_2}} \right) \right) \right] \\\nonumber
    &\leq \prod_{i \in [n-1]} \left( \Pr \left[ \bigwedge_{h_1 \in \{1,...,\lfloor k/2 \rfloor\}} \overline{ \EE_{i,h_1}} \right] + \Pr \left[ \bigwedge_{h_2 \in \{\lfloor k/2 \rfloor+1,...,k\}} \overline{ \EE_{i,h_2}}  \right] \right) \\\label{eq:calD}
    &= \prod_{i \in [n-1]} \left( \prod_{h_1 \in \{1,...,\lfloor k/2 \rfloor\}} \Pr\left[\overline{\EE_{i,h_1}}\right] + \prod_{h_2 \in \{\lfloor k/2 \rfloor+1,...,k\}} \Pr\left[\overline{\EE_{i,h_2}}\right]
    \right).
\end{align}
The second inequality follows by applying the union bound. The last equality follows since for a given $i\in [n-1]$, the events in $\left\{\EE_{i,h}: h\in \{1, ..., \lfloor k/2\rfloor\}\right\}$ (and, respectively, the events $\left\{\EE_{i,h}: h\in \left\{\lfloor k/2\rfloor+1, ..., k\right\}\right\}$) are mutually independent.
Using the upper bound for $\Pr\left[\overline{\EE_{i,h}}\right]$ from Lemma~\ref{lemma:actualbound}, we finally get 
\begin{align*}
    \Pr[\calD] &< \left(2 \cdot\exp\left(-\frac{k-1}{2} \cdot \exp(-48)\right)\right)^{n-1} \\
    &\leq \exp \left( (n-1) \cdot \left( -\frac{k-1}{2} \cdot \exp(-48) + 1 \right) \right) \\
    &\leq \exp \left( -(n-1) \cdot \frac{k}{3} \cdot \exp(-48) \right) \\
    &\leq \exp(-m \ln{k})=k^{-m},
\end{align*}
as desired. The first inequality follows by equation (\ref{eq:calD}) after observing that each of the sets $\left\{1, ..., \lfloor k/2\rfloor\right\}$ and $\left\{\lfloor k/2\rfloor+1, ..., k\right\}$ have at least $\frac{k-1}{2}$ elements. The second inequality is obvious, the third one follows by the condition $k\geq 3+6e^{48}$, and the fourth one by the definition of $m$.
\end{proof}

Using a result of~\citet[Theorem~3.1]{MS22} which lower-bounds \cd$(n,k)$ by \disc$(n,k)$, we obtain the following corollary.

\begin{theorem}\label{thm:cd}
Let $k\geq 3+6e^{48}$ and $n\geq 1+147\cdot e^{48} \ln{k}$ be integers. Then, there is a set of $n$ agents with valuations for a set of items so that no $k$-allocation of the items is \emph{\cd}$d$ for some $d\in \Omega\left(\sqrt{\frac{n}{\ln{k}}}\right)$.
\end{theorem}

\section{Lower bounds for envy-freeness and proportionality}\label{sec:ef-prop}

In this section, we prove our lower bounds for $\ef(n_1, ..., n_k)$ (Theorem~\ref{thm:ef}) and $\prop(n_1, ..., n_k)$ (Theorem~\ref{thm:prop} and Theorem~\ref{thm:propnew}).

\begin{theorem}\label{thm:ef}
    Let $k\geq 2$ be an integer and $n\geq k+242\cdot e^{124}k\ln{k}$. Then, for every set of $n$ agents, partitioned into $k$ non-empty groups, there is a set of items and valuations of the agents for these items so that no allocation is \emph{\ef}$d$ for some $d\in \Omega\left(\sqrt{\frac{n}{k\ln{k}}}\right)$.
\end{theorem}

\begin{theorem}\label{thm:prop}
    Let $k\geq 2$ be an integer and $n\geq k+162\cdot e^{77}k\ln{k}$. Then, for every set of $n$ agents, partitioned into $k$ non-empty groups, there is a set of items and valuations of the agents for these items so that no allocation is \emph{\prop}$d$ for some $d\in\Omega\left( \sqrt{\frac{n}{k^3 \ln{k}}}  \right)$.
\end{theorem}

\begin{theorem} \label{thm:propnew}
For every set of agents partitioned into $k\geq 4$ groups so that the number of agents in each group is at least $1+32e^{96}k^2\ln{k}$, there is a set of items and valuations of the agents for these items so that no allocation is \emph{\prop}$d$ for some $d\in\Omega\left( \sqrt{\frac{\min(n_1,n_2,\dots,n_k)}{\ln{k}}} \right)$.
\end{theorem}

We prove Theorems~\ref{thm:ef}, ~\ref{thm:prop} and ~\ref{thm:propnew} using a probabilistic argument that adapts the one we used for multi-color discrepancy in Section~\ref{sec:disc}. For appropriate parameters $m$ and $d$, we will define a valuation profile in which the agent valuations for $m$ items are random and will show that the probability that there exists an allocation that is envy-free/proportional up to $d$ items is strictly smaller than $1$. This implies that there exists an instance that does not admit an allocation that is envy-free/proportional up to $d$ items. 

For each group $h\in [k]$, we select a specific agent to be the {\em leader} of the group. We denote by $L$ the set of agents who are group leaders and by $F$ the remaining agents. The instance has a set $M$ of $m$ items. For each item $j \in M$, the group leader $i\in L$ has valuation $v_i(j)=1$ for the item. The valuation $v_i(j)$ of an agent $i\in F$ for item $j\in M$ is decided by tossing a fair coin; it is equal to $1$ on heads (this happens with probability $1/2$) and equal to $0$ on tails. All coin tosses are independent.

\subsection{Proof of Theorem~\ref{thm:ef}}
We prove Theorem \ref{thm:ef} using the above construction with $m=\left\lfloor \frac{n-k}{2e^{124}\ln{k}}\right\rfloor$ and set $d=\sqrt{\frac{m}{k}}$. Together with the restriction $n\geq k+242\cdot e^{124}k\ln{k}$, these definitions imply that $d\in \Omega\left(\sqrt{\frac{n}{k\ln{k}}}\right)$. Furthermore, notice that $\sqrt{m/k}\geq 11$ and $m\geq 11kd$; these observations will be useful later.

Denote by $\A$ the set of allocations $A=(A_1, A_2, ..., A_k)$ of the items in $M$ to $k$ bundles so that
\begin{align}\label{eq:ef-bundle-size-bounds}
    \frac{m}{k}-d &\leq |A_h| \leq \frac{m}{k}+d
\end{align}
for every $h\in [k]$. Clearly, an allocation that does not belong to $\A$ is not envy-free up to $d$; the condition for \ef$d$ would be violated for some agent in $L$. We will show that the random valuation profile is envy-free up to $d$ for an allocation $A\in \A$ is less than $k^{-m}$. Since there are at most $k^m$ allocations in $\A$, the probability that some of them is envy-free up to $d$ will be strictly less than $1$, completing the proof of Theorem~\ref{thm:ef}.

For an agent $i\in F$ and allocation $A\in \A$, we denote by $\EE_i(A)$ the event defined as
\begin{align}
    v_i(A_{g(i)}) &\geq v_i(A_h)-d
\end{align}
for every $h\in [k]$. The following lemma provides the main argument in our proof.

\begin{lemma}\label{lem:ef-bound}
    Consider an allocation $A\in \A$ and agent $i\in F$. Then,
    \begin{align*}
        \Pr[\EE_i(A)] &< \exp\left(-\frac{1}{2}\cdot \exp\left(-124\right)\right).
    \end{align*}
\end{lemma}

\begin{proof}
We will focus on two additional events $\EE_i^1(A)$ and $\EE_i^2(A)$. Let $h^*$ be an arbitrary group in $[k]\setminus \{g(i)\}$. Event $\EE_i^1(A)$ is true if and only if
\begin{align}
    v_i(A_{g(i)}) &\leq \frac{|A_{h^*}|}{2}-\frac{3d}{2}.
\end{align}
Event $\EE_i^2(A)$ is true if and only if
\begin{align}
    v_i(A_{h^*}) &\geq \frac{|A_{h^*}|}{2}.
\end{align}
Set $\eps=\frac{5d}{|A_{g(i)}|}$ and notice that 
\begin{align}\nonumber
    \Pr\left[\EE_i^1(A)\right] &= \Pr\left[v_i(A_{g(i)}) \leq \frac{|A_{h^*}|}{2}-\frac{3d}{2}\right] \geq \Pr\left[v_i(A_{g(i)})\leq \frac{|A_{g(i)}|}{2}-\frac{5d}{2}\right]\\\label{eq:ef-1}
    &= \Pr\left[v_i(A_{g(i)})\leq \frac{|A_{g(i)}|}{2} (1-\eps)\right].
\end{align}
The inequality follows by equations (\ref{eq:ef-bundle-size-bounds}), which imply that $|A_{h^*}|\geq |A_{g(i)}|-2d$.

Now, notice that the random variable $v_i(A_{g(i)})$ follows the binomial probability distribution $\bin(|A_{g(i)}|,1/2)$ with $|A_{g(i)}|$ trials (one for each item in bundle $A_{g(i)}$) and success probability $1/2$ per trial. We will bound the RHS of inequality (\ref{eq:ef-1}) by applying the reverse Chernoff bound (Lemma~\ref{lem:reverse-Chernoff}) with  $\eps=\frac{5d}{|A_{g(i)}|}$ and $t = |A_{g(i)}|$. The next claim justifies that we can indeed do so.

\begin{claim}{\label{claim:use-reverse-Chernoff}}
For $\eps=\frac{5d}{|A_{g(i)}|}$, it holds $\eps \leq 1/2$ and $\eps^2\cdot |A_{g(i)}|\geq 6$.
\end{claim}

\begin{proof}
Recall that $m\geq 11dk$ and $d=\sqrt{m/k}$. Thus, by the left part of equation (\ref{eq:ef-bundle-size-bounds}), we have $|A_{g(i)}|\geq \frac{m}{k}-d\geq 10d$ and $\eps=\frac{5d}{|A_{g(i)}|}\leq 1/2$. Furthermore, by the right part of equation (\ref{eq:ef-bundle-size-bounds}), we have $|A_{g(i)}|\leq \frac{m}{k}+d\leq \frac{12m/k}{11}$ and $\eps^2 |A_{g(i)}|=\frac{25m}{k|A_{g(i)}|}\geq 275/12 > 6$.
\end{proof}

Thus, by applying Lemma~\ref{lem:reverse-Chernoff}, inequality (\ref{eq:ef-1}) yields
\begin{align}\nonumber
    \Pr\left[\EE_i^1(A)\right] &\geq \exp\left(-\frac{9\eps^2 |A_{g(i)}|}{2}\right)=\exp\left(-\frac{225d^2}{2|A_{g(i)}|}\right) \geq \exp\left(-\frac{225d^2}{2(m/k-d)}\right)\\\label{eq:ef-2}
    &= \exp\left(-\frac{225\sqrt{m/k}}{2(\sqrt{m/k}-1)}\right)\geq \exp\left(-\frac{2375}{20}\right)> \exp(-124).
\end{align}
The first inequality follows by the left part of equation (\ref{eq:ef-bundle-size-bounds}) and the second one follows by the fact $\sqrt{m/k}\geq 11$.

The random variable $v_i(A_{h^*})$ follows the binomial probability distribution $\bin(|A_{h^*}|,1/2)$ with $|A_{h^*}|$ trials and success probability $1/2$ per trial. Thus,
\begin{align}\label{eq:ef-3}
    \Pr\left[\EE_i^2(A)\right] &\geq 1/2.
\end{align}
Clearly, since the bundles $A_{g(i)}$ and $A_{h^*}$ are disjoint, the valuations of agent $i$ for the two bundles are independent. Thus, by inequalities (\ref{eq:ef-2}) and (\ref{eq:ef-3}), we get
\begin{align}\label{eq:ef-4}
    \Pr\left[\EE_i^1(A)\wedge \EE_i^2(A)\right] &> \frac{1}{2}\cdot \exp\left(-124\right).
\end{align}
Now, observe that event $\EE_{i}(A)$ is true only if some of events $\EE_i^1(A)$ and $\EE_i^2(A)$ is false.  Hence,
\begin{align*}
    \Pr[\EE_i(A)] &\leq 1-\Pr\left[\EE_i^1(A)\wedge \EE_i^2(A)\right]< 1-\frac{1}{2}\cdot \exp\left(-124\right) \leq \exp\left(-\frac{1}{2}\cdot \exp\left(-124\right)\right)
\end{align*}
using inequality (\ref{eq:ef-4}) and the property $1-x\leq e^{-x}$ for every real $x$.
\end{proof}

Now, allocation $A\in \A$ is envy-free up to $d$ if the event $\EE_i(A)$ is true for every agent $i\in F$. These events are independent as each of them depends on the random valuations of a different agent in $F$. Thus, by Lemma~\ref{lem:ef-bound} and since $|F|=n-k$, we have
\begin{align*}
    \Pr\left[\bigwedge_{i\in F}{\EE_i(A)}\right] &< \exp\left(-\frac{n-k}{2}\cdot \exp\left(-124\right)\right)\leq \exp\left(-m\ln{k}\right)=k^{-m},
\end{align*}
as desired. The second inequality follows by the definition of $m$.

\subsection{Proof of Theorem~\ref{thm:prop}}
To prove Theorem~\ref{thm:prop}, we use our construction above with $m=\left\lfloor \frac{n-k}{2e^{77}\ln{k}}\right\rfloor$ and set $d=\sqrt{\frac{m}{k^3}}$. Together with the restriction $n\geq k+162e^{77}k\ln{k}$, these definitions imply that $d\in \Omega\left(\sqrt{\frac{n}{k^3 \ln{k}}}\right)$ and, furthermore, $m\geq 9k^2 d$ (this fact will be useful later).

Denote by $\A$ the set of allocations $A=(A_1, A_2, ..., A_k)$ of the items in $M$ to $k$ bundles so that 
\begin{align}\label{eq:prop-def-A}
|A_h| &\geq \frac{m}{k}-d 
\end{align}
for every $h\in [k]$. Clearly, an allocation that does not belong to $\A$ is not proportional up to $d$ items, as the condition for \prop$d$ would be violated for some agent in $L$. We will show that the probability that the random valuation profile is proportional up to $d$ items for an allocation $A\in \A$ is at most $k^{-m}$. Since there are fewer than $k^m$ allocations in $\A$, the probability that some of them is proportional up to $d$ items will be strictly less than $1$, completing the proof.

Let $A\in \A$. For an agent $i\in F$, we denote by $\EE_i(A)$ the event defined as 
\begin{align}\label{eq:prop-event-def}
    v_i(A_{g(i)}) &\geq \frac{1}{k}\cdot \sum_{h\in [k]}{v_i(A_h)}-d.
\end{align}
The following lemma provides the main argument in our proof.
\begin{lemma}\label{lem:prop-bound}
Consider the allocation $A\in \A$ and agent $i\in F$. Then, 
\begin{align*}
    \Pr[\EE_i(A)] &\leq \exp\left(-\frac{1}{2}\cdot \exp\left(-77\right)\right).
\end{align*}
\end{lemma}

\begin{proof}
Since the allocation $A$ belongs to $\A$, inequality (\ref{eq:prop-def-A}) implies
\begin{align}\label{eq:prop-upper-for-cal-A}
    |A_{g(i)}| &=m-\sum_{h\in [k]\setminus \{g(i)\}}{|A_h|}\leq \frac{m}{k}+(k-1)d.
\end{align}

We will argue about two additional events $\EE_i^1(A)$ and $\EE_i^2(A)$. Event $\EE_i^1(A)$ is true if and only if 
\begin{align}\label{eq:prop-1}
    v_i(A_{g(i)}) &\leq \frac{|A_{g(i)}|}{2}-2kd.
\end{align}
Event $\EE_i^2(A)$ is true if and only if
\begin{align}\label{eq:prop-2}
    \sum_{h\in [k]\setminus\{g(i)\}}{v_i(A_h)} &\geq \frac{m-|A_{g(i)}|}{2},
\end{align}

Set $\eps=\frac{4kd}{|A_{g(i)}|}$ and observe that
\begin{align}\label{eq:prop-use-reverse-Chernoff}
    \Pr[\EE_i^1(A)] & =\Pr\left[v_i(A_{g(i)}) \leq \frac{|A_{g(i)}|}{2}-2kd\right] = \Pr\left[v_i(A_{g(i)}) \leq \frac{|A_{g(i)}|}{2}\cdot (1-\eps)\right].
\end{align}
Now, notice that, for agent $i\in F$, the random variable $v_i(A_{g(i)})$ follows the binomial probability distribution $\bin(|A_{g(i)}|,1/2)$ with $|A_{g(i)}|$ trials (one for each item of bundle $A_{g(i)}$) and success probability $1/2$ per trial. We will bound the RHS of equation (\ref{eq:prop-use-reverse-Chernoff}) by applying the reverse Chernoff bound (Lemma~\ref{lem:reverse-Chernoff}) with $\eps=\frac{4kd}{|A_{g(i)}|}$ and $t = |A_{g(i)}|$. The next claim justifies that we can indeed do so.

\begin{claim}
    For $\eps=\frac{4kd}{|A_{g(i)}|}$ and $t = |A_{g(i)}|$, it holds that $\eps\leq 1/2$ and $\eps^2|A_{g(i)}| \geq 6$.
\end{claim}

\begin{proof}
Recall that $m\geq 9k^2d$ and $d=\sqrt{m/k^3}$. Thus, by inequality (\ref{eq:prop-def-A}), we have $|A_{g(i)}| \geq \frac{m}{k}-d\geq 8kd$ and $\eps=\frac{4kd}{|A_{g(i)}|}\leq 1/2$. Furthermore, by inequality (\ref{eq:prop-upper-for-cal-A}), we have $|A_{g(i)}|\leq \frac{m}{k}+(k-1)d\leq \frac{m}{k}+kd\leq \frac{10m}{9k}$ and $\eps^2t=\frac{16k^2d^2}{|A_{g(i)}|}\geq \frac{72k^3d^2}{5m}\geq 6$.
\end{proof}

Using inequality (\ref{eq:prop-def-A}) and the facts $m\geq 9k^2d$ and $k\geq 2$, we have
\begin{align}\label{eq:prop-4}
    |A_{g(i)}| &\geq \frac{m}{k}-d \geq \frac{m}{k}-\frac{m}{9k^2} \geq \frac{17m}{18k}.
\end{align}
Now, by applying Lemma~\ref{lem:reverse-Chernoff} and using inequality (\ref{eq:prop-4}), we have
\begin{align}
    \Pr[\EE_i^1(A)] & \geq \exp\left(-\frac{9\eps^2|A_{g(i)}|}{2}\right) = \exp\left(-\frac{72k^2d^2}{|A_{g(i)}|}\right) \geq \exp\left(-\frac{1296k^3d^2}{17m}\right)>\exp(-77).
\end{align}

Also, the random variable $\sum_{h\in [k]\setminus\{g(i)\}}{v_i(A_h)}$ follows the binomial probability distribution $\bin(m-|A_{g(i)}|,1/2)$ with $m-|A_{g(i)}|$ trials (one for each item not belonging to bundle $A_{g(i)}$) and success probability $1/2$ per trial. Thus, 
\begin{align}
    \Pr[\EE_i^2(A)] &\geq \frac{1}{2}.
\end{align}
Notice that the valuations of agent $i$ for the items in $A_{g(i)}$ and $\cup_{h\in [k]\setminus\{g(i)\}}{A_h}$ are selected independently and, hence, the events $\EE_i^1(A)$ and $\EE_i^2(A)$ are independent. Thus, 
\begin{align}
    \Pr[\EE_i^1(A) \wedge \EE_i^2(A)] &\geq \frac{1}{2}\exp\left(-77\right).
\end{align}

We now claim that if the events $\EE_i^1(A)$ and $\EE_i^2(A)$ are true, then the event $\EE_i(A)$ is false. Indeed, using the inequalities (\ref{eq:prop-1}) and (\ref{eq:prop-2}), we obtain that
\begin{align}\nonumber
    \sum_{h\in [k]\setminus \{g(i)\}}{v_i(A_h)}-(k-1)\cdot v_i(A_{g(i)})
    &\geq \frac{m-|A_{g(i)}|}{2}-(k-1)\cdot \left(\frac{|A_{g(i)}|}{2}-2kd\right)\\\nonumber
    &= \frac{m}{2}-\frac{k\cdot |A_{g(i)}|}{2} +2k(k-1)d\\\nonumber
    &\geq \frac{m}{2}-\frac{k}{2}\cdot \left(\frac{m}{k}+(k-1)d\right)+2k(k-1)d\\\label{eq:prop-3}
    &= \frac{3k(k-1)}{2}d > kd.
\end{align}
The second inequality follows by inequality (\ref{eq:prop-upper-for-cal-A}) and the third one by the facts $k\geq 2$ and $d>0$. Equivalently, equation (\ref{eq:prop-3}) yields 
\begin{align*}
    v_i(A_{g(i)}) &< \frac{1}{k}\cdot \sum_{h\in [k]\setminus \{g(i)\}}{v_i(A_h)}-d,
\end{align*}
contradicting the definition of event $\EE_i(A)$ in Equation~(\ref{eq:prop-event-def}). Thus, event $\EE_i(A)$ is true only if some of $\EE_i^1(A)$ and $\EE_i^2(A)$ is not true, i.e., 
\begin{align*}
    \Pr[\EE_i(A)] &\leq 1-\Pr[\EE_i^1(A)\wedge \EE_i^2(A)]\\
    &\leq 1- \frac{1}{2}\exp\left(-77\right)\\
    &\leq \exp\left(-\frac{1}{2}\exp\left(-77\right)\right).
\end{align*}
The last inequality follows by the property $1-x\leq e^x$ for every real $x$. The proof of the lemma is complete.
\end{proof}

Now, allocation $A\in \A$ is proportional up to $d$ items if the event $\EE_i(A)$ is true for every agent $i\in F$. These events are independent as each of them depends on the random valuations of a different agent. Thus, by Lemma~\ref{lem:prop-bound} and since $|F|=n-k$, we have
\begin{align*}
    \Pr\left[\bigwedge_{i\in F}{\EE_i(A)}\right] &< \exp\left(-\frac{n-k}{2}\exp\left(-77\right)\right)\leq  \exp\left(-m\ln{k}\right)=k^{-m},
\end{align*}
as desired. The second inequality follows by the definition of $m$.

\subsection{Proof of Theorem~\ref{thm:propnew}}
We prove Theorem~\ref{thm:propnew} using our construction with $m=\left\lfloor\frac{(\min\{n_1,n_2, ..., n_k\}-1)k}{2e^{96}\ln{k}} \right\rfloor$ and set $d=\sqrt{m/k}$. Together with the assumption $\min\{n_1,n_2, ..., n_k\}\geq 1+32e^{96}k^2\ln{k}$, this implies $d\leq \frac{m}{4k^2}$ (and, obviously, $d\leq \frac{m}{4k}$); these facts will be useful later in the proof. 

Again, denote by $\A$ the set of allocations $A=(A_1,A_2, ..., A_k)$ of the items in $M$ to $k$ bundles so that $|A_h|\geq m/k-d$ for every $h\in [k]$. As argued in the proof of Theorem~\ref{thm:prop}, no other allocation can be proportional up to $d$ items.

Let $A\in \A$. For an agent $i\in F$, we denote by $\EE_i(A)$ the event defined as 
\begin{align}\label{eq:event-propnew-e_i-of-A}
v_i(A_{g(i)}) &\geq \frac{1}{k}\cdot \sum_{h\in [k]}{v_i(A_h)}-d.
\end{align}

For $h\in [k]$, let $\zeta_h\geq 0$ be such that 
\begin{align}
|A_h| &= \frac{m}{k}+\left(\zeta_h-1\right)d.
\end{align}
Clearly, since $\sum_{h\in [k]}{|A_h|}=m$, it holds 
\begin{align}\label{eq:sum-of-zetas}
\sum_{h\in [k]}{\zeta_h}&=k.
\end{align}
Hence, $\zeta_h\leq k$ for every $h\in [k]$. 

The next lemma is crucial for our proof.
\begin{lemma}\label{lem:propnew-bound}
    Consider the allocation $A\in \A$ and agent $i\in F$. Then,
\begin{align*}
    \Pr[\EE_i(A)] &< \exp\left(-\frac{1}{2}\exp\left(-\frac{6\left(3+\zeta_{g(i)}\right)^2d^2k}{m}\right)\right)
\end{align*}
\end{lemma}

\begin{proof}
    We will argue about two additional events $\EE_i^1(A)$ and $\EE_i^2(A)$. Event $\EE_i^1(A)$ is true if and only if 
\begin{align}\label{eq:propnew-event-1}
    v_i(A_{g(i)}) &\leq \frac{m}{2k}-2d,
\end{align}
while event $\EE_i^2(A)$ is true if and only if
\begin{align}\label{eq:propnew-event-2}
    \sum_{h\in [k]\setminus\{g(i)\}}{v_i(A_h)} &\geq \frac{m(k-1)}{2k}-\frac{(k-1)d}{2}.
\end{align}
Using the definition of $\zeta_{g(i)}$ and setting $\eps=\frac{(3+\zeta_{g(i)})d}{|A_{g(i)}|}$, we have
\begin{align}\nonumber
    \Pr\left[\EE_i^1(A)\right] &=\Pr\left[v_{i(A_{g(i)})} \leq \frac{m}{2k}-2d\right]\\\nonumber
    &=\Pr\left[v_i(A_{g(i)}) \leq \frac{|A_{g(i)}|}{2}-\frac{3+\zeta_{g(i)}}{2}d\right]\\\label{eq:propnew-use-reverse-Chernoff}
    &=\Pr\left[v_i(A_{g(i)}) \leq \frac{|A_{g(i)}|}{2}(1-\eps)\right].
\end{align}
Now, notice that, for agent $i\in F$, the random variable $v_i(A_{g(i)})$ follows the binomial probability distribution $\bin(|A_{g(i)}|,1/2)$ with $|A_{g(i)}|$ trials (one for each item of bundle $A_{g(i)}$) and success probability $1/2$ per trial. We will bound the RHS of equation (\ref{eq:propnew-use-reverse-Chernoff}) by applying the reverse Chernoff bound (Lemma~\ref{lem:reverse-Chernoff}) with $\eps=\frac{(3+\zeta_{g(i)})d}{|A_{g(i)}|}$ and $t=|A_{g(i)}|$. The next claim justifies that we can indeed do so.

\begin{claim}
    For $\eps=\frac{(3+\zeta_{g(i)})d}{|A_{g(i)}|}$ and $t=|A_{g(i)}|$, it holds that $\eps\leq 1/2$ and $\eps^2t \geq 6$.
\end{claim}

\begin{proof}
By the definition of $\zeta_{g(i)}$, we have $\eps=\frac{(3+\zeta_{g(i)})d}{|A_{g(i)}|}= \frac{(3+\zeta_{g(i)})d}{m/k+(\zeta_{g(i)}-1)d}$, which is non-decreasing for $\zeta_{g(i)}\in [0,k]$ since $d\leq \frac{m}{4k}$. Thus, since $k\geq 4$ and $d\leq \frac{m}{4k^2}$, we get $\eps\leq \frac{(3+k)d}{m/k+(k-1)d}\leq \frac{2k^2d}{m}\leq \frac{1}{2}$.

Also, by the definition of $\zeta_{g(i)}$, we have $\eps^2 t = \frac{(3+\zeta_{g(i)})^2d^2}{\frac{m}{k} + (\zeta_{g(i)}-1)d}$. The derivative of this expression w.r.t $\zeta_{g(i)}$ has the same sign as the
quantity 
\begin{align*}
    &2(3+\zeta_{g(i)})\left(\frac{m}{k} + (\zeta_{g(i)}-1)d\right) - (3 + \zeta_{g(i)})^2d\\
    &= (3 + \zeta_{g(i)}) \left( \frac{2m}{k} + d(\zeta_{g(i)} - 5\right)
    \geq \left(3 + \zeta_{g(i)}\right)\left(\frac{2m}{k} - 5d\right) \geq 0.
\end{align*}
Thus, it is minimized for $\zeta_{g(i)} = 0$ and, hence, $\eps^2t \geq \frac{9d^2}{m/k-d} \geq \frac{12d^2k}{m} \geq 6$,
as desired. The second last inequality follows by the fact $d\leq \frac{m}{4k}$ and the last inequality follows since  $d=\sqrt{m/k}$, by definition.
\end{proof}

We get 
\begin{align}\nonumber
    \Pr\left[\EE_i^1(A)\right] &\geq \exp\left(-\frac{9(3+\zeta_{g(i)})^2d^2}{2|A_{g(i)}|}\right)=\exp\left(-\frac{9(3+\zeta_{g(i)})^2d^2}{2(m/k+(\zeta_{g(i)}-1)d}\right)\\\label{eq:propnew-event-1-prob}
    &\geq \exp\left(-\frac{6(3+\zeta_{g(i)})^2d^2k}{m}\right).
\end{align}
Furthermore,
\begin{align}\nonumber
    \Pr\left[\EE_i^2(A)\right] &= \Pr\left[\sum_{h\in [k]\setminus \{g(i)\}}{v_i(A_h)} \geq \frac{m(k-1)}{2k}-\frac{(k-1)d}{2}\right]\\\nonumber
    &=\Pr\left[\sum_{h\in [k]\setminus\{g(i)\}}{v_i(A_h)} \geq \sum_{h\in [k]\setminus\{g(i)\}}{\frac{|A_h|-(\zeta_h-1)d}{2}}-\frac{k-1}{2}d\right]\\\label{eq:propnew-event-2-prob}
    &\geq \Pr\left[\sum_{h\in [k]\setminus\{g(i)\}}{v_i(A_h)} \geq \frac{1}{2}\sum_{h\in [k]\setminus\{g(i)\}}{|A_h|}\right]\geq \frac{1}{2}.
\end{align}
The second last inequality follows since $\zeta_{g(i)}\leq k$. The last inequality follows since the random variable $\sum_{h\in [k]\setminus\{g(i)\}}{v_i(A_h)}$ follows the binomial probability distribution $\bin\left(\sum_{h\in [k]\setminus\{g(i)\}}{|A_h|},1/2\right)$ and the LHS is just the probability that this random variable is at least its mean.

We now claim that event $\EE_i(A)$ is false when events $\EE_i^1(A)$ and $\EE_i^2(A)$ are true. Indeed, assuming that inequalities (\ref{eq:propnew-event-1}) and (\ref{eq:propnew-event-2}) are true, we obtain
\begin{align}\nonumber
    \frac{1}{k}\sum_{h\in [k]}{v_i(A_h)}-v_i(A_{g(i)}) &=\frac{1}{k}\sum_{h\in [k]\setminus\{g(i)\}}{v_i(A_h)}-\frac{k-1}{k}v_i(A_{g(i)}) \\\nonumber
    &\geq \frac{m(k-1)}{2k^2}-\frac{(k-1)d}{2k}-\frac{m(k-1)}{2k^2}+\frac{2(k-1)}{k}d\\\nonumber
    &\geq \frac{3(k-1)}{2k}d> d.
\end{align}
The last inequality follows since $k\geq 4$.

Now, observe that the events $\EE_i^1(A)$ and $\EE_i^2(A)$ are independent as they refer to the valuation of agent $i\in F$ for disjoint sets of items. Hence, using equations (\ref{eq:propnew-event-1-prob}) and (\ref{eq:propnew-event-2-prob}), we get
\begin{align*}
    \Pr\left[\EE_i(A)\right] &\leq 1-\Pr\left[\EE_i^1(A)\right]\cdot \Pr\left[\EE_i^2(A)\right]\\
    &\leq \exp\left(-\Pr\left[\EE_i^1(A)\right]\cdot \Pr\left[\EE_i^2(A)\right]\right)\\
    &\leq  \exp\left(-\frac{1}{2}\exp\left(-\frac{6(3+\zeta_{g(i)})^2d^2k}{m}\right)\right).
 \end{align*}
This completes the proof of the lemma.\end{proof}

Let $\ell=\min\{n_1, n_2, ..., n_k\}-1$. For $h\in [k]$, denote by $F_h$ a set of $\ell$ non-leader agents from group $h$. Let $\widetilde{F}=\cup_{h\in [k]}{F_h}$. 

Now, observe that allocation $A\in \A$ is proportional up to $d$ items only if the event $\EE_i(A)$ is true for every agent $i\in \widetilde{F}$. These events are independent as each of them depends on the random valuations of a different agent. Thus, by Lemma~\ref{lem:propnew-bound}, we have
\begin{align}\nonumber
\Pr\left[\bigwedge_{i\in \widetilde{F}}{\EE_i(A)}\right] &< \prod_{i\in \widetilde{F}}{\exp\left(-\frac{1}{2}\exp\left(- \frac{6\left(3+\zeta_{g(i)}\right)^2d^2k}{m}\right)\right)}\\\nonumber
&=\prod_{h\in [k]}{\prod_{i\in F_h}{\exp\left(-\frac{1}{2}\exp\left(- \frac{6\left(3+\zeta_{g(i)}\right)^2d^2k}{m}\right)\right)}}\\\nonumber
&=\prod_{h\in [k]}{\exp\left(-\frac{\ell}{2}\exp\left(- \frac{6\left(3+\zeta_{h}\right)^2d^2k}{m}\right)\right)}\\\label{eq:use-jensen}
&=\exp\left(-\frac{\ell}{2}\sum_{h\in [k]}{\exp\left(- \frac{6\left(3+\zeta_h\right)^2d^2k}{m}\right)}\right)
\end{align}
We can easily verify that the function $\exp(-c\cdot (3+z)^2)$ is convex in $[0,+\infty)$ for $c\geq 1/18$. Then, using Jensen's inequality and equation (\ref{eq:sum-of-zetas}), we have 
\begin{align}\label{eq:convexity}
\sum_{h\in [k]}{\exp\left(-c\cdot (3+\zeta_h)^2\right)} &\geq \exp\left(-c\cdot \left(3+\frac{1}{k}\sum_{h\in [k]}{\zeta_h}\right)^2\right)=\exp\left(-16c\right).
\end{align}
Using equation (\ref{eq:convexity}) for $c=\frac{6d^2k}{m}=6$, equation (\ref{eq:use-jensen}) yields
\begin{align*}
    \Pr\left[\bigwedge_{i\in \widetilde{F}}{\EE_i(A)}\right] &<\exp\left(-\frac{k\ell}{2}\exp\left(-\frac{96d^2k}{m}\right)\right)=\exp\left(-\frac{k\left(\min\{n_1,n_2,...,n_k\}-1\right)}{2e^{96}}\right)\\
    &\leq \exp\left(-m\ln{k}\right)=k^{-m},
\end{align*}
as desired. The first equality follows by the definition of $d$ and $\ell$ and the second inequality follows by the definition of $m$.

\section{Conclusion}\label{sec:open}
We have presented an improved lower bound on multi-color discrepancy and improved lower bounds on how much we should relax the fairness notions of consensus division, envy-freeness, and proportionality when sets of indivisible items have to be allocated in $k$ groups of $n$ agents in total. There is still a small gap of $\Theta\left(\sqrt{\ln{k}}\right)$ from the currently known upper bounds for multi-color discrepancy and consensus division, and larger gaps for the other two fair division properties. Closing these gaps are the obvious open problems that stem from our work. Exploring additional implications of our multi-color discrepancy lower bound is an interesting task as well. We suspect that it will find applications to areas completely different than fair division, which has been our focus here.

\section*{Acknowledgements}
Ioannis Caragiannis and Sudarshan Shyam were partially supported by the Independent Research Fund Denmark (DFF) under grant 2032-00185B.
Kasper Green Larsen is co-funded by a DFF Sapere Aude Research Leader Grant No. 9064-00068B by the Independent Research Fund Denmark and co-funded by the European Union (ERC, TUCLA, 101125203). Views and opinions expressed are however those of the author(s) only and do not necessarily reflect those of the European Union or the European Research Council. Neither the European Union nor the granting authority can be held responsible for them.

\bibliographystyle{ACM-Reference-Format}
\bibliography{references}

\end{document}